\documentclass[letterpaper,8pt]{extarticle}  
\usepackage{import}
\usepackage{local}
\usepackage[left=.65in,top=.65in,bottom=.65in,right=.65in]{geometry}
\usepackage{lipsum} 
\usepackage{algorithm2e}
\RestyleAlgo{ruled}
\SetKwFor{For}{for}{do}{endfor} 
\usepackage{booktabs}
\usepackage{xcolor}
\usepackage{soul}
\usepackage{ulem}
\usepackage{amsthm}
\usepackage{hyperref}
\hypersetup{colorlinks,citecolor=Accent, linkcolor=Accent, urlcolor=Accent}

\usepackage{xspace}

\newcommand{\mycomment}[1]{}

\newcommand{\sysname}{SPIDER\xspace}
\newcommand{\basename}{baseSPIDER\xspace}

\newtheorem{theorem}{Theorem}
\newtheorem{lemma}[theorem]{Lemma}

\title{SPIDER: Two Server Functionality for the Cost of Zero}

\author{%
\textbf{%
Ofir Dvir\orcidlink{0009-0001-7418-6392}\textcolor{Accent}{\textsuperscript{1}}\,\href{mailto:odvir@ucsb.edu}{\small odvir@ucsb.edu},\space
Kali Hale\orcidlink{0009-0001-2925-0607}\textcolor{Accent}{\textsuperscript{1}}\,\href{mailto:kalihale@ucsb.edu}{\small kalihale@ucsb.edu},\space
Javin Zipkin\orcidlink{0009-0004-7010-657X}\textcolor{Accent}{\textsuperscript{1}}\,\href{mailto:javinzipkin@ucsb.edu}{\small javinzipkin@ucsb.edu},\space
Divyakant Agrawal\orcidlink{0000-0002-0215-9539}\textcolor{Accent}{\textsuperscript{1}}\,\href{mailto:divyagrawal@ucsb.edu}{\small divyagrawal@ucsb.edu},\space
Dahlia Malkhi\orcidlink{0000-0002-7038-7250}\textcolor{Accent}{\textsuperscript{1}}\,\href{mailto:dahliamalkhi@ucsb.edu}{\small dahliamalkhi@ucsb.edu}}\\
\begin{small}
\textcolor{Accent}{\textsuperscript{1}}University of California, Santa Barbara \\
\end{small}
}

\date{}
\begin{document}
\maketitle
\thispagestyle{empty}


\begin{doublespacing}
\section*{Abstract}
\noindent
\textbf{\textcolor{Accent}{We introduce baseSPIDER and SPIDER, private information retrieval (PIR) schemes that embody two technical advancements.

The baseSPIDER protocol operates with a single server and a stateful client that performs pre-processing and stores hints for future queries. In this setting, baseSPIDER introduces a new approach that matches the asymptotically optimal communication complexity of state-of-the-art schemes while improving constant factors--an advantage that is particularly significant for databases with large entries. In addition, baseSPIDER offers a conceptually simpler design relative to prior protocols.

SPIDER operates over a \textit{default} database interface and requires no cooperation from the server at any stage. To our knowledge, SPIDER is the first single-server PIR construction of this design, achieving privacy without specialized APIs, auxiliary server state, or protocol-specific interaction beyond conventional indexed access.

SPIDER is built via a simple transformation of baseSPIDER to the default server setting, eliminating deployment barriers and enabling immediate applicability to existing systems. This transformation can be applied more broadly to three recent PIR solutions, adapting them for use in the default-server paradigm and yielding solutions of independent interest. SPIDER compares to the resulting modified solutions by exhibiting a simpler design while incurring higher client computational work.
}}

\section{Introduction}

On the open web, data providers often keep detailed profiles of individual users. These profiles may include websites visited, location history, app data, address books, nearby wireless devices, cellular networks, WiFi networks, and more. Coupled with the \textit{specific} content a user has retrieved from a website, this clearly violates \textit{privacy}, because profile information may be easily correlated with specific users, may be used to glean insights and analytics into businesses, or may simply used for targeted advertising.
In other words, the fact that the information has been gathered at all is a privacy risk, regardless of the intended use.

A classical technique known as Private Information Retrieval (PIR) enables a client to retrieve an element from a remote database without revealing which element is accessed. There exists a rich line of prior work on PIR in various settings. 
Recent advances bring PIR into the practical realm by operating with a single server and a stateful client that performs pre-processing and stores hints for future queries, avoiding the use of any heavy cryptographic computations on the server-side.
Building on and complementing these advances, this paper introduces \basename and \sysname\footnote{SPIDER stands for Simple Private Information DEfault-server Retrieval.}, a pair of PIR schemes that address two distinct remaining objectives for deploying PIR in the web browsing setting.

\medskip

The first goal is scaling to the large volumes of content being accessed and delivered from web servers. 
Modern web servers routinely serve large objects, ranging from kilobyte-scale records such as patent entries, to megabyte-scale images, and even gigabyte-scale video content. At the same time, these systems must sustain high levels of concurrency, servicing many client requests simultaneously. In such environments, even modest per-query overheads can quickly compound. Consequently, download bandwidth in PIR is not merely a theoretical metric but a critical systems parameter, and achieving practical performance requires improving even constant factors, particularly in bandwidth-intensive workloads.

Addressing this objective, \basename embodies a new and practical approach in the single-server PIR settings, that matches the asymptotically optimal communication complexity of~\cite{RenLing2024SaPA} while improving constant factors. This improvement becomes particularly consequential for databases with large entries. In addition, \basename offers a conceptually simpler design than prior protocols. 

Our second objective is to eliminate the need for a specialized server-side API for privacy-preserving functionality. Like most recent single-server PIR constructions, \basename requires the server to perform simple operations---e.g., XORing retrieved items---before returning a response. However, even such minimal assumptions may fail in the increasingly relevant setting where clients access \emph{default} web services. In this environment, servers are uncooperative, cannot (and lack incentives to) be modified, and expose only standard read-only, index-based access to their content. Indeed, providing fine-grained privacy-preserving functionality may directly conflict with such services’ business models, monitoring objectives, or abuse-prevention policies. Consequently, practical PIR systems targeting the open web must assume a \emph{default} server: one that performs no additional computation, maintains no PIR-specific metadata, and behaves indistinguishably from a standard public content server.
    
Addressing this objective, \sysname adapts \basename into a PIR scheme that assumes only read-only access to the server’s content and requires \emph{no cooperation} beyond standard indexed retrieval. To our knowledge, \sysname is the first solution to operate in the default-server setting. Moreover, we observe that this transformation extends more broadly to a class of single-server PIR constructions. In particular, when instantiated with prior cooperative single-server schemes~\cite{IEEESP:10646686_PIANO_2024_5,RenLing2024SaPA}, it yields new protocols that preserve (or nearly preserve) their efficiency.

This transformation requires the client to download multiple entries per query, and thus \sysname does not retain the optimal per-query communication achieved by \basename. As a result, our two objectives are not simultaneously satisfied, and combining them remains an open problem. On the other hand, retrieving multiple entries in full yields an additional benefit: all schemes adapted to this setting (as well as those that natively support it) enable uninterrupted, continuous querying\footnote{PIANO-23 refers to this property as ``unbounded''.}. After a one-time setup, query refresh can be incorporated into the standard query flow, without requiring additional interaction or communication beyond what is already inherent to the protocol.

\paragraph{Background}

In order to introduce our solutions, we first give a brief background on previous PIR advances that it builds upon.
We assume a single honest-but-curious server that serves content correctly but attempts to infer client interest based on observed requests. The server may record all access patterns, correlate them with external context, throttle suspicious traffic, or apply fingerprinting techniques. The client seeks \emph{access-pattern privacy}, meaning that the server should not learn which database index the client is interested in. Unlike settings where anonymizing networks (e.g., Tor \cite{torproject2026}) obfuscate client identity, our goal is orthogonal: even if the server knows \emph{exactly which client} is querying, it should learn nothing about \emph{which entry} the client retrieves.

A single-server must touch all entries, otherwise it learns the client is \textbf{not} interested in some entries~\cite{10.1007/3-540-44598-6_4}.
However, practical considerations rule out the server using any cryptographic protocols, such as homomorphic encryption, to query the server while hiding the content of the query \textit{from} the server, since they incur a prohibitive server computation cost. Hence we are in the Information Theoretic PIR (IT-PIR) setting. 

Patel, Persiano, and Yeo ~\cite{10.1145/3243734.3243821} pioneered a single-server IT-PIR approach that consists of two parts, preprocessing and retrieval. 
The preprocessing stage streams the entire database and stores hints on the client, hence it is called \textit{stateful}. 
Assuming the server holds $n$ entries of size $\beta$, the preprocessing stage results in $O(n\cdot \beta)$ communication.
The goal of preprocessing is to enable recurring retrievals from the server privately without streaming the entire database each time.
A follow up stateful PIR work by Corrigan-Gibbs and Kogan ~\cite{cryptoeprint:2019/1075_CGK} demonstrated the first stateful PIR scheme with amortized sublinear per-retrieval complexity of $O(n^{3/4}\cdot \beta)$.

A series of recent advancements tighten these results and bring practical, stateful single-server PIR solutions. Of these, the most relevant ones here are two seminal schemes due to Zhou et al.~\cite{cryptoeprint:2023/452_PIANO_11_12, IEEESP:10646686_PIANO_2024_5}, referred to as ``PIANO-23'' and ``PIANO-24'', resp.\, 
and Ren et al.\cite{RenLing2024SaPA}, referred to as ``RMS-24'' (we are specifically concerned with the single-server variant presented in this paper). 

The PIANO methods introduce a hint construction strategy that brings client-side storage and simultaneously per-query communication down $\Tilde{O}(\sqrt{n}\cdot\beta)$, the optimal (neglecting logarithmic factors). 
The hint construction of PIANO is simple and elegant. Briefly, the client partitions the database into $\sqrt{n}$ ``chunks''. Each hint is an XOR of $\sqrt{n}$ elements, one selected at random from each chunk. 
When querying, the client replaces the desired entry in the vector it requests from the server with a stored entry from the same chunk, thus hiding the actual query from the server. The client then XORs the response from the server with the stored hint and the replacement entry to reveal the desired entry.

While the PIANO methods do not support arbitrary query workloads and depart from the standard PIR model, RMS-24 present a single-server solution that follows a similar hint approach but removes this constraint. It halves the size of hints to $\sqrt{n}/2 + 1$, thus slightly increasing client storage (by factor $2$). For retrieval, the client removes the desired entry from the request vector, then generates a ``dummy'' hint that covers the remaining $\sqrt{n}/2$ chunks (including the chunk containing the desired entry). The server then sends back an XOR for each hint, one for the real hint and one for the dummy hint, thus reducing retrieval complexity to $2\cdot\beta$.

\paragraph{Our results.}

The core of our new approach is a new hint construction, which is not based on sharding and is generally simpler than previous hint schemes.
The construction uses multisets of $\sqrt{n}$ elements selected uniformly at random among all $\sqrt{n}$-multisets. 
The key insight is that if any individual element is dropped from a multiset, the remaining $(\sqrt{n}-1)$-multiset is distributed uniformly and independently of the dropped element among all $(\sqrt{n}-1)$-multisets. We refer to such a $(\sqrt{n}-1)$-multiset as ``redacted''. Therefore, a redacted $(\sqrt{n}-1)$-multiset does not reveal any information about the removed element. 

Employing this idea, for the preprocessing stage the client samples multisets uniformly at random among all $\sqrt{n}$-multisets.
A sampled multiset serves as a \textit{hint} as follows: the client stores a pair of values, a succinct representation of the indices that make up the multiset and an XOR of the actual entries in it. 
As usual, the number of hints needs to be $\Tilde{O}(\sqrt{n})$ in order to achieve coverage of all entries with high probability.

In each retrieval step, the client finds a hint that contains the desired index, redacts it, and queries the server for the redacted multiset.
The server XORs the redacted hint and returns a single value to the client. The client then XORs the retrieved value with the stored hint in order to retrieve the desired element.

\basename achieves the following result in the cooperative-server setting:

\begin{theorem}
    There exists an IT-PIR scheme with per-query communication of \textbf{exactly one word of size $\beta$}, amortized communication of $\sqrt{n}\cdot\beta$, client storage of $O(\sqrt{n} \cdot \ln{n} \cdot \beta)$, and amortized server computation of $\sqrt{n}$.
\end{theorem}

\noindent Note that this improves retrieval compared to the best known method from $2\cdot\beta$ to $1\cdot\beta$, a substantive improvement when database entries are sizable. \basename also uses a different hint construction technique.

\paragraph{Working with Default Servers}

All of the above solutions assume a cooperative server. Using a simple transformation, we modify \basename's retrieval step to utilize only standard retrieve-by-index API on the server. 
Using this transformation, we obtain \sysname, a solution in the default server setting in which each query retrieves $\sqrt{n}$ database entries from server. In \sysname, the client retrieves from the server the (actual) entries corresponding to the remaining indices of the redacted hint, instead of their XOR. The client then XORs the retrieved entries with the stored hint in order to retrieve the desired element.

\begin{theorem}
    There exists an IT-PIR scheme with amortized communication of $\sqrt{n}$, client storage of $O(\sqrt{n} \cdot \log{n} \cdot \beta)$ and amortized server computation of $\sqrt{n}$, that requires no PIR-specific protocols to be carried out by the server.
\end{theorem}

Finally, since $\sqrt{n}$ actual database entries are retrieved in each query, \sysname uses them to continuously refresh hints. In this way, essentially no communication beyond what is already inherent to the protocol is needed to refresh hints.

\medskip
In summary, \basename and \sysname achieve the following properties:

\begin{itemize}
    \item \textbf{Small client storage:}
    The client stores $O(\sqrt{n}\log n)$ hints, amounting to
    $O(\sqrt{n}\cdot\log n\cdot \beta)$ space.
    \item \textbf{Sub-linear (amortized) communication:}
    Each query requires requesting exactly $1$ (\basename) and $\sqrt{n}$ (\sysname) database elements.
    \item \textbf{Default server:}
    In \sysname, the server executes \emph{no} PIR-specific protocol beyond retrieving the specific
    indices supplied by the client. There are no proxies, no shuffle-net, and no helper parties.
    \item \textbf{Linear pre-processing with continuous use:}
    In \sysname, the client streams the database once and constructs the hints during this pass. The client then proceeds with continuous querying, refreshing hints from the entries it obtains when querying.
\end{itemize}

\section{Model \& Definitions}

\label{sec:ModelAndDef}

\begin{table}[t]
    \centering
    \small
        \begin{tabular}{|c|p{0.72\linewidth}|}
            \hline
            \textbf{Notation} & \textbf{Meaning} \\ \hline
            \multicolumn{2}{|c|}{\textbf{--- Database ---}} \\ \hline 
            \verb|DB| & Denotes the database. \\ \hline
            $n$ & Number of database entries. \\ \hline
            $\beta$ & Size of the largest entry in the database. \\ \hline
            \multicolumn{2}{|c|}{\textbf{--- Hints ---}} \\ \hline
            $R$ & A single sampled hint \\ \hline
            $k$ & Size of each hint, usually taken to be $k=\sqrt{n}$ \\ \hline
            $m$ and $C$ & Number of stored hints $m$, where $m = 2\cdot C \cdot \ln{n} \cdot (n/k)$.\\ \hline
            $[n]$ & Universe of database indices, i.e.\ $\{1,\dots,n\}$. \\ \hline
            $x$ and $y$ & Entries in a hint primarily used in the context of \textbf{coverage}. \\ \hline
            $i$ and $j$ & Entries in a hint primarily used in the context of \textbf{privacy}. \\ \hline
        \end{tabular}
    \caption{Notation Table for the paper.}
    \label{tab:notation}
\end{table}

In this section, we formalize the problem of private information retrieval (PIR) in the single-server sublinear information theoretic (IT-PIR) setting. At a high level, single-server PIR allows a client to perform a preprocessing phase on a database and then privately perform a number of queries on the database. IT-PIR requires no cryptographic techniques such as homomorphic encryption, necessitating only XORs on the part of the client and server. We complete this formalization in ~\ref{subsec:pir}.

We adopt a semi-honest, or honest-but-curious, threat model, in which any adversary can observe all communication and make any attempt to infer the user's intent, but makes no attempt to disrupt or alter the protocol. The complete formalization of this trust model is in \ref{subsec:trust}.

\subsection{Private Information Retrieval (PIR)}
\label{subsec:pir}

A database \verb|DB| of $n$ entries is assumed, where the entries are identified by a unique identifier, typically an index or a key. The largest entry in the database is size $\beta$. In each query, the client wishes to retrieve the entry identified by $i$ from the database. A PIR scheme must ensure:

\begin{itemize}
    \item \textbf{Correctness}: The client can retrieve the entry requested each time.
    \item \textbf{Privacy}: The server learns nothing about the entry the client has retrieved.
\end{itemize}

These definitions may also be formalized as a game between the client, who wishes to keep their requests private, and an adversary who observes the server.

\begin{enumerate}
    \item The adversary has, through observation, narrowed the client's next choice down to one of two entries, identified by $i$ and $j$.
    \item The client flips a coin to choose one of the two entries. If the coin lands with the face side up, the client will retrieve $i$. Otherwise, the client will retrieve $j$.
    \item The adversary gains no information as to which entry the client may have chosen.
\end{enumerate}

A single-server sublinear IT-PIR scheme requires a preprocessing phase \cite{cryptoeprint:2024/976}, in which the client will stream the entire database before using the preprocessed database to query individual entries at a lower cost. The preprocessing phase must therefore not reveal any information indicating which entry or entries the client intends to retrieve - indeed, the preprocessing phase must allow \textit{any} entry to be retrieved.

Repeated queries must also not reveal any information about the desired entries (including repeated queries of the same entry), which necessarily requires that, in a single-server PIR scheme, the server must be unaware of the preprocessed hints \textit{and} the same hint can never be used twice (since the likelihood of two hints containing the exact same entries with one difference is highly unlikely, using the same hint will reveal the item that was retrieved). This can be achieved using a \textit{replenishment} scheme to ensure that the entries accessed with a used hint can still be accessed without giving away any additional information about that particular hint.

\subsection{Trust Model}
\label{subsec:trust}

This paper considers a trust model in which the server is \textbf{semi-honest} (also referred to as honest-but-curious). The server faithfully stores the database and correctly responds to all retrieval requests, but it may attempt to infer the identity of the queried entry identifier from the queries it observes. The server does not deviate from the protocol - it does not send malformed responses, abort, or selectively refuse to answer. However, it retains and may analyze all queries it receives.

Additionally, \textbf{all communications between the client and server are assumed to be public}: any message transmitted during the protocol - including queries, responses, and any preprocessing traffic - is observable by a passive adversary. Under this model, privacy requires that the observable transcript of the protocol reveals no information about the entry identifier $i$ being queried, even to an adversary who sees the entirety of the communication.

Finally, the entries on the database server must be static.

\section{Protocol \& Algorithms}
\label{sec:solution}

This section details the \basename protocol in five parts. We describe the preprocessing phase, in which the client derives a compact hint-set from the database which will become the basis for online queries, and requirements of the hint construction in \ref{subsec:preproc}. Continuing in \ref{subsec:query}, we detail how the client chooses a hint and constructs a request to the server which is used to retrieve a desired entry. Using a hint to retrieve an entry \textit{consumes} that hint, and so the protocol must then replace that hint with another suitably constructed hint, as detailed in \ref{subsec:replenishment}. In order to minimize communication requirements, an entry can be \textit{cached} for the duration of the query phase as detailed in \ref{subsec:entryCaching}, such that any repeated requests are retrieved from client storage. Finally, in \ref{subsec:continuous_preproc}, we describe an improvement which allows \textit{continuous preprocessing}, eliminating the reqirement for a distinct preprocessing phase.

\subsection{Preprocessing \& Hint Construction}
\label{subsec:preproc}

Preprocessing commences when the client streams the database and assembles the
hint set. In order to maintain privacy in the default-server setting, the
client must stream the entire database, so the communication overhead of this phase is
$O(n\beta)$.

During the streaming pass, the client constructs $m$ hints. Each hint $R$ is a
size-$k$ multiset of database indices together with the XOR of the
corresponding database entries. While the specific numerical choices for $k$ and $m$ are discussed in
Section~\ref{sec:MathAnalysis}, intuitively, $k$ controls the per-query online
cost, while $m$ is chosen large enough so that every database entry is covered
by at least one hint with high probability.

Rather than storing the full list of $k$ indices for each hint, the client stores each hint compactly as a single 64-bit seed together with its
precomputed $\beta$-byte XOR value. The seed serves as a succinct identifier for the multiset: given the same seed, Algorithm~\ref{alg:expand-multiset-seed} deterministically reconstructs the
same size-$k$ multiset whenever the hint is needed.

The multiset expansion proceeds by a direct stars-and-bars bijection. Let $N = n+k-1$.
The client first uses a seeded pseudorandom generator to sample a uniformly random $k$-subset
\[
S = \{u_1,\dots,u_k\} \subseteq [N],
\qquad u_1 < \cdots < u_k.
\]
It then converts this subset into a size-$k$ multiset over $[n]$ by setting
\[
h_t = u_t-(t-1)
\qquad \text{for } t=1,\dots,k.
\]
This subtraction removes the positional offsets that were introduced to make repeated multiset elements distinct. Equivalently, the forward map sends a multiset
\[
1 \le h_1 \le h_2 \le \cdots \le h_k \le n
\]
to the strictly increasing subset
\[
u_t = h_t+(t-1)
\qquad \text{for } t=1,\dots,k.
\]
For example, the multiset $(2,2,3)$ corresponds to the subset $(2,3,5)$:
the second copy of $2$ is shifted upward by one position, and the $3$ is
shifted upward by two positions. Applying the inverse map recovers the original
multiset as
\[
2,\quad 3-1=2,\quad 5-2=3.
\]

Because this correspondence is bijective, sampling a uniform $k$-subset of $[n+k-1]$ is equivalent to sampling a uniform size-$k$ multiset over $[n]$.
Thus, the persistent representation of a hint consists only of its seed and its stored XOR value. This significantly reduces client storage while preserving the ability to reconstruct the exact hint on demand. A full accounting of
storage, together with the derivation of the coverage guarantee, is deferred to
Section~\ref{sec:MathAnalysis}.

\subsection{Online Query}
\label{subsec:query}

To retrieve the database entry at index $i$, the client must select a hint that contains $i$, use that hint to reconstruct the entry, and do so in a way that reveals nothing to the server about the value of $i$.

The client first selects a hint $R_i$ that covers the target index $i$. It then expands the seed of $R_i$ using Algorithm~\ref{alg:expand-multiset-seed} to recover the underlying size-$k$ multiset, removes one copy of $i$ from that multiset, and sends the remaining $k-1$ indices to the server. This operation is referred to as \emph{redaction}. Algorithm~\ref{alg:PuncHint} describes the corresponding retrieval procedure. Retrieval from the server proceeds as follows:

\begin{description}
    \item[\basename.] In the cooperative setting, the server returns the XOR of the requested entries.

    \item[\sysname.] In the default-server setting, the server returns the requested entries individually, and the client performs the XOR locally. 

\end{description} 

In either scheme, since the stored hint value was originally computed as the XOR of all $k$ entries in the multiset, and the requested multiset contains the remaining $k-1$ entries after redaction, the final XOR reveals exactly the desired entry at index $i$.

\subsection{Hint Replenishment}
\label{subsec:replenishment}

Each hint is single-use, meaning that once a hint has been used in a query, it is discarded and never selected again. This prevents the server from being able to use the transcript of all previous requests to correlate multiple redacted queries derived from the same underlying multiset. The privacy guarantee then follows from the symmetry of the multiset construction: when one copy of the queried index is removed from a uniformly chosen covering hint, the resulting redacted multiset is distributed independently of the queried index. The server therefore observes a query that is statistically indistinguishable from a query for any other target. The formal privacy proof is deferred to Section~\ref{sec:MathAnalysis}.

A naive implementation of the single-use rule would simply discard each consumed hint and leave the remaining pool untouched. However, as noted in~\cite{IEEESP:10646686_PIANO_2024_5}, simply removing a used hint introduces a subtle bias. Over time, removing only the consumed hints would skew the remaining hint set toward indices the client has not yet queried, which leaks information about the client's access pattern through the residual structure of the hint pool.

Grounded in previous approaches~\cite{IEEESP:10646686_PIANO_2024_5, cryptoeprint:2023/452_PIANO_11_12, RenLing2024SaPA}, to counteract this drift each hint $R$ is augmented during preprocessing with a single \emph{replacement slot}: one of its $k$ multiset elements, denoted $p$, is designated as replaceable and stored explicitly. When a hint $R_i$ covering index $i$ is consumed and removed from the active pool, the client picks a uniformly random surviving hint $R_x$ and rewrites its replacement element from $p$ to $i$, updating the stored XOR value; $i$ is added to $R_x$ by XORing the $i$ with $R_x$, and similarly, $p$ is removed the same way so that the updated hint multiset is $(R_x \uplus \{i\}) \setminus \{p\}$.

\subsection{Entry Caching}
\label{subsec:entryCaching}

Once the client has queried index $i$, it stores the value of the entry at $i$ locally until the next hint refreshing phase (i.e. until $\sqrt{n}$ queries have been completed). Thus, any subsequent request for $i$ is then answered directly from local storage. This both eliminates redundant work and avoids the need to ever cover the same index twice in one phase.

As shown in the coverage analysis (Section ~\ref{subsec:mincoverage}), by choosing the number $m$ of hints appropriately, the hints remaining in the pool after  $\sqrt{n}$ queries suffice to cover non-queried indices with high probability. After that, the hint pool is refreshed.

\subsection{Continuous Pre-Processing}
\label{subsec:continuous_preproc}

The initial hint set supports only a bounded number of online queries, since each hint is single-use and is discarded once consumed. To sustain a longer query stream, preprocessing can instead be carried out continuously in the background, so that fresh hints are generated as older ones are used.

This is possible because each hint consists only of a seed together with the XOR of the corresponding database entries. In \sysname, the online phase already returns the raw entries needed to reconstruct redacted hints: in the
default-server setting, each query causes the server to return the $k-1$ entries of the redacted multiset individually. These downloaded entries may
therefore be reused not only to recover the current target, but also to help
assemble new hints in the background.

Concretely, the client may continuously sample fresh seeds, expand them via
Algorithm~\ref{alg:expand-multiset-seed}, and maintain partial XORs for the resulting multisets as matching entries arrive through normal online queries. Once a certain number of queries have been sent and we have enough information, we refresh the hint multiset, throwing away all old hints and finishing off any pre-processing steps left. In this way, hint generation becomes an ongoing local maintenance task rather than a separate
one-shot phase.

\begin{figure}
    \centering
    \includegraphics[alt={A diagram illustrating a Private Information Retrieval (PIR) or similar database query model. A database (DB) is shown as a row of n indexed cells (0, 1, 2, … n−1). Arrows point to a central user node, indicating two types of database access: reading the entire DB once, or reading k−1 words per query. The user node also has an arrow pointing downward to a storage symbol (a stack of disks), labeled "store m words once", representing local storage the user maintains after a preprocessing phase.}, width=0.5\linewidth]{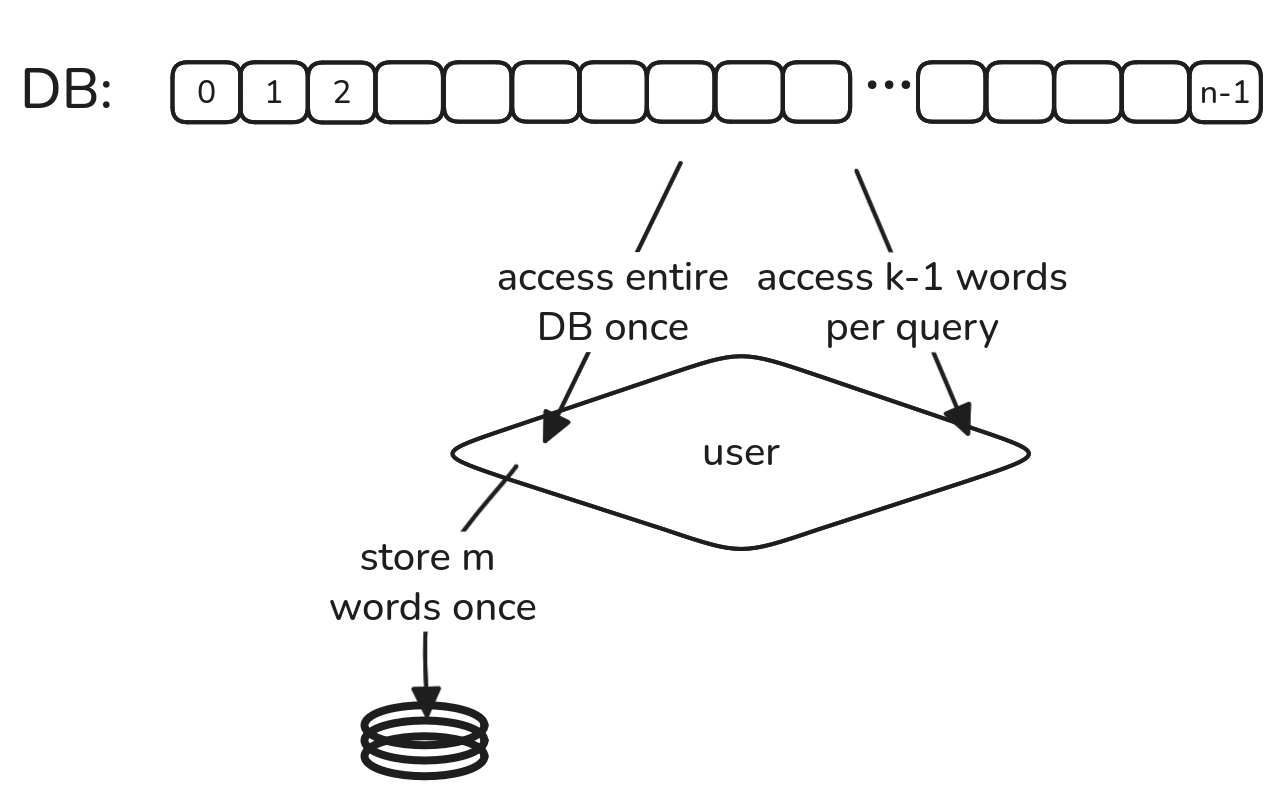}
    \caption{Communication and storage complexity for the default-server scheme.}
    \label{fig:placeholder}
\end{figure}

\begin{algorithm}
    \caption{Redacted Hint Retrieval}
    \label{alg:PuncHint}

    \KwIn{Redacted list of entries to be retrieved from database}
    \KwResult{Redacted hint to be returned to user}

    \tcc{`retrieve(i)` is a function which retrieves the $i$th entry from the database, while `redact` is the redacted list of entries to be sent to the database.}

    result = retrieve(redact[0])\; 
    
    \For{i = 1 \KwTo redact.size - 1}{
    
        result = result $\oplus$ retrieve(redact[i])\;
        
    }

    return result\;
\end{algorithm}

\begin{algorithm}
    \caption{\textsc{ExpandMultisetFromSeed}$(n,k,\mathit{seed})$}
    \label{alg:expand-multiset-seed}

    \KwIn{Universe size $n$, multiset size $k$, 64-bit seed $\mathit{seed}$}
    \KwResult{A uniformly random size-$k$ multiset $R=(h_1,\dots,h_k)$ over $[n]$}
    
    $N \gets n+k-1$\;
    
    \tcc{Sample a uniformly random $k$-subset $S \subseteq [N]$ using a seeded PRG}
    $S \gets \textsc{FloydSample}(N,k,\mathit{seed})$\;
    
    \tcc{Sort $S$ in ascending order so that $u_1 < \cdots < u_k$}
    $(u_1,\dots,u_k) \gets \textsc{Sort}(S)$\;
    
    \tcc{Apply the inverse stars-and-bars map}
    \For{$t \gets 1$ \KwTo $k$}{
        $h_t \gets u_t - (t-1)$\;
    }
    
    \Return $(h_1,\dots,h_k)$\;
\end{algorithm}

\section{Correctness \& Privacy Analysis}
\label{sec:MathAnalysis}
\begin{table}[t]
    \centering
    \small
        \begin{tabular}{|c|p{0.72\linewidth}|}
        \hline
        \multicolumn{2}{|c|}{\textbf{--- Probabilities ---}} \\ \hline
            $p$ & Probability that a fixed entry appears in a uniformly sampled hint:
            \[
            p = \frac{k}{n+k-1}.
            \]
            \\ \hline
            $M$ & Total number of size-$k$ multisets over $[n]$:
            \[
            M = \binom{n+k-1}{k}.
            \]
            \\ \hline
            $S_y$ & Number of size-$k$ multisets over $[n]$ that contain a fixed entry $y$.
            \\ \hline
            $Y_y$ & Number of hints, out of the $m$ total sampled hints, that contain entry $y$. \\ \hline
            $\delta$ & Slack parameter in the lower-tail Chernoff bound, used to show $Y_y \ge (1-\delta)\mathbb{E}[Y_y]$ with high probability. \\ \hline
        \end{tabular}
    \caption{Notation Table for the paper.}
    \label{tab:correctness_notation}
\end{table}

In this section, we present proofs for correctness and privacy of our construction. We define correctness as the ability to retrieve any item with high probability, bounded by a failure probability $\delta$. Minimum coverage is detailed in \ref{subsec:mincoverage}, while the coverage required to eliminate a replenishment requirement in \ref{subsec:intendedCoverage}. We then discuss how \basename satisfies privacy requirements of privacy and indistinguishability, ensuring that the intended target is indistinguishable from a query of any random item, in \ref{subsec:privacy}.

\subsection{Requirements for Minimum Coverage}
\label{subsec:mincoverage}

Let $Y_y$ denote the number of sampled hints, out of the $m$ total hints,
that contain $y$. Since \sysname samples hints uniformly without replacement from
the full multiset space, $Y_y$ is hypergeometric. More explicitly, letting
\[
M := \binom{n+k-1}{k}
\]
denote the total number of size-$k$ multisets over $[n]$, and
\[
S_y := \binom{n+k-2}{k-1}
\]
denote the number of such multisets that contain $y$, we have
\[
Y_y \sim \mathrm{Hypergeometric}(M,S_y,m).
\]
In particular,
\[
\mathbb{E}[Y_y]
=
m\frac{S_y}{M}
=
m\frac{k}{n+k-1}.
\]

This immediately gives the basic coverage bound. Defining
\[
X_y := \mathbf{1}\{Y_y=0\},
\]
the event $X_y=1$ occurs exactly when all $m$ sampled hints come from the
$M-S_y$ multisets that avoid $y$. Hence
\[
\Pr[X_y=1]
=
\frac{\binom{M-S_y}{m}}{\binom{M}{m}}.
\]
Therefore, by linearity of expectation,
\[
\mathbb{E}\!\left[\sum_{y=1}^n X_y\right]
=
n\cdot \frac{\binom{M-S_y}{m}}{\binom{M}{m}}.
\]
Applying Markov's inequality gives
\[
\Pr\!\left[\sum_{y=1}^n X_y>0\right]
\le
n\cdot \frac{\binom{M-S_y}{m}}{\binom{M}{m}}.
\]
Thus, to achieve full coverage with failure probability at most $\delta$, it
suffices that
\[
n\cdot \frac{\binom{M-S_y}{m}}{\binom{M}{m}} \le \delta.
\]

\subsection{Requirements for Intended Coverage} 
\label{subsec:intendedCoverage}

Having established how many hints are needed to ensure that every entry is covered at least once, the next step is to strengthen the requirement and ask for a larger initial amount of coverage per entry. This is needed so that \sysname can answer a sequence of queries without requiring a replenishment scheme. 

Suppose $m$ is chosen so that
\[
\mathbb{E}[Y_y] = 2C\ln n
\]
for some constant $C>0$, i.e.,
\[
m = \frac{2C\ln n}{p}.
\]
The goal is then to show that, for sufficiently large constant $C$, every entry is covered by a constant fraction of its expectation with high probability.
For a fixed $y$, a standard lower-tail Chernoff bound for the hypergeometric distribution gives that for any fixed constant $\delta\in(0,1)$,
\[
\Pr\!\left[Y_y \le (1-\delta)\mathbb{E}[Y_y]\right]
\le
\exp\!\left(-\frac{\delta^2\mathbb{E}[Y_y]}{2}\right).
\]
Substituting $\mathbb{E}[Y_y]=2C\ln n$ yields
\[
\Pr\!\left[Y_y \le (1-\delta)2C\ln n\right]
\le
\exp\!\left(-\frac{\delta^2\cdot 2C\ln n}{2}\right)
=
n^{-\delta^2 C}.
\]
Applying a union bound over all $y\in[n]$ gives
\[
\Pr\!\left[\exists y\in[n] : Y_y \le (1-\delta)2C\ln n\right]
\le
n\cdot n^{-\delta^2 C}
=
n^{\,1-\delta^2 C}.
\]
Thus, whenever
\[
\delta^2 C > 1,
\]
it follows that
\[
\min_{y\in[n]} Y_y \ge (1-\delta)2C\ln n
\qquad\text{w.h.p.}
\]

\paragraph{Correctness} For correctness, we must show that there is a high likelihood of the user being able to retrieve the required entry. In our case, we show that the likelihood of not finding a candidate hint for a specific entry is approximately $e^{-\alpha}$, making the number of uncovered entries about $ne^{-\alpha}$ (per \ref{subsec:mincoverage}). A simple solution to ensure access to all entries in the database would be to simply store any entries that are not used to construct a hint, as the number of entries that will not be in a hint is relatively low (ideally, not more than $m$ entries such that the storage required is no more than $2m$).

\subsection{Privacy Analysis} 
\label{subsec:privacy}

Per Section~\ref{subsec:pir}, privacy requires that the server cannot distinguish between the case in which the queried entry is $i$ and the case in which the queried entry is $j$. Equivalently, for any two candidate targets
$i,j \in [n]$, the observable transcript of the protocol must have the same distribution in both cases.

Before proving privacy for the actual construction, it is helpful to explain
why the use of multisets is essential. The discussion below first shows that an ordered-hint variant would be insecure, and then explains why the multiset representation removes exactly this asymmetry.

\noindent\textbf{Stage 1 (why ordered hints fail).}
Suppose the server observes exactly one copy of entry $j$ in the redacted hint. If the queried target was $j$, then the original hint must have contained two copies of $j$, so that one copy remained after redaction. If instead the
queried target was some $i \neq j$, then the same observation could arise from an original hint containing one copy of $i$ and one copy of $j$, with the copy of $i$ being removed.

If hints are sampled as ordered tuples, these two cases do not occur with the same probability. For example, when $k=2$, the ordered hints $(i,j)$ and
$(j,i)$ are distinct, whereas $(j,j)$ is the only ordered hint containing two copies of $j$. Thus, after observing a single visible copy of $j$, the server can distinguish between the cases ``the client redacted $i$'' and ``the client redacted $j$'' with nonzero advantage. In other words, an ordered-hint variant would not satisfy privacy.

\noindent\textbf{Stage 2 (why multisets fix the problem).}
This asymmetry disappears when hints are sampled as multisets. Indeed, fix two distinct entries $i,j \in [n]$. The number of size-$k$ multisets containing one copy each of $i$ and $j$ is equal to the number of size-$k$ multisets containing two copies of $j$ and no copy of $i$.

In the first case, one copy of $i$ and one copy of $j$ are fixed, so the remaining $k-2$ elements may be chosen as any multiset over the remaining
$n-2$ entries. This gives
\[
\binom{(n-2)+(k-2)-1}{k-2}
\]
possibilities. In the second case, two copies of $j$ are fixed and $i$ is excluded, so again the remaining $k-2$ elements may be chosen as any multiset over the same $n-2$ entries, yielding the same count,
\[
\binom{(n-2)+(k-2)-1}{k-2}.
\]
Thus, unlike in the ordered case, observing one visible copy of $j$ does not favor one queried target over the other.

This symmetry is precisely what underlies the privacy of the multiset construction. The formal proof below strengthens this intuition by showing that, for every redacted size-$(k-1)$ multiset $P$ and every queried target $i$, there is exactly one size-$k$ multiset containing $i$ that redacts to $P$, namely $P \uplus \{i\}$.

\paragraph{Formal privacy proof.}
Privacy for the multiset construction follows from a simple symmetry argument. Once one copy of the queried index is removed, the resulting redacted multiset is distributed independently of the queried index.

Let $M_{k-1}$ denote the set of all size-$(k-1)$ multisets over $[n]$.
Recall that
\[
|M_{k-1}|=\binom{n+k-2}{k-1}.
\]
\begin{lemma}[Single-query redaction hides the target]
\label{lem:single-query-privacy}
Fix any target index $i \in [n]$. Let $R_i$ denote the set of all size-$k$
multisets over $[n]$ that contain at least one copy of $i$. Sample a hint $R$
uniformly from $R_i$, and remove one copy of $i$. Let $P$ denote the resulting redacted multiset of size $k-1$.

Then $P$ is distributed uniformly over $M_{k-1}$.
\end{lemma}

\begin{proof}
Fix any redacted multiset $P \in M_{k-1}$. There is exactly one multiset in
$R_i$ that produces $P$ after removing one copy of $i$, namely
\[
H = P \uplus \{i\},
\]
where $\uplus$ denotes multiset union. Thus the map
\[
P \mapsto P \uplus \{i\}
\]
is a bijection from $M_{k-1}$ to $R_i$. Therefore,
\[
|R_i| = |M_{k-1}| = \binom{n+k-2}{k-1}.
\]
Since $R$ is sampled uniformly from $R_i$, every $P \in M_{k-1}$ occurs with
probability
\[
\Pr[P=P_0 \mid \text{target } i]
=
\frac{1}{|R_i|}
=
\frac{1}{\binom{n+k-2}{k-1}}.
\]
Hence $P$ is uniform over $M_{k-1}$.
\end{proof}

Lemma~\ref{lem:single-query-privacy} immediately yields one-query privacy: for
any two candidate targets $i,j \in [n]$, the server observes exactly the same
distribution on the redacted multiset.

\paragraph{Multi-query transcript privacy.}
Assume that, after each query, the client's local update rule preserves the invariant that any fresh hint later used for target $i$ is distributed
uniformly over $R_i$. Then for any two target sequences
\[
(i_1,\dots,i_Q), (j_1,\dots,j_Q)\in[n]^Q,
\]
the induced distributions on the full server transcript
\[
(P_1,\dots,P_Q)
\]
are identical.

\begin{proof}
Privacy must also hold over a sequence of queries. The argument above shows that in any individual round, redacting a fresh hint hides the queried target.
To extend this to repeated queries, it remains to show that the distribution of the fresh hints used in later rounds is preserved over time.

After a query for index $i$, the consumed hint is discarded. The client may then update another randomly selected stored hint locally by substituting $i$ for one selected element, together with the corresponding XOR update. This
local update is not visible to the server. By assumption, it preserves the invariant that, whenever a later query uses a fresh hint containing some target $i$, that hint is distributed uniformly over $R_i$.

Accordingly, at every query round $t$, if the queried target is $i_t$, the client uses a fresh size-$k$ hint distributed uniformly over $R_{i_t}$,
removes one copy of $i_t$, and sends the resulting redacted multiset $P_t$.
By Lemma~\ref{lem:single-query-privacy}, $P_t$ is uniform over $M_{k-1}$ and therefore independent of $i_t$.

Now consider any sequence of $Q$ client queries with targets
\[
(i_1,\dots,i_Q)\in[n]^Q,
\]
and let $(P_1,\dots,P_Q)$ denote the redacted multisets observed by the server. For each round $t$, the conditional distribution of $P_t$ given any previously observed transcript prefix $(P_1,\dots,P_{t-1})$ is uniform over
$M_{k-1}$ and hence does not depend on the queried target $i_t$. Therefore, for any two target sequences
\[
(i_1,\dots,i_Q), (j_1,\dots,j_Q)\in[n]^Q,
\]
the induced joint distributions on $(P_1,\dots,P_Q)$ are identical. Thus the
full server transcript reveals no information about which indices were queried.
\end{proof}

\section{A Default-Server Setting Transformation}
\label{sec:discussion}

As evident in Section~\ref{sec:solution}, \basename and \sysname are very similar. \sysname is derived from \basename by having the client obtain the actual entries of a redacted hint and XORing them, rather then asking the server to XOR the retrieved entries. 

The reason this transformation is possible is because in the cooperative paradigm, \basename already has an amortized communication of $O(\sqrt{n})$: 
Although the result of a single query returned from the server is a single XORed result, the preprocessing phase (performed every $\sqrt{n}$ queries) requires communication of $n$. 
This leads to a simple conclusion: If the client retrieves $\sqrt{n}$ entries per query, and each hint contains $\sqrt{n}$ entries, the server does not need to XOR the result of the redacted hint --- the user can simply retrieve the $\sqrt{n}$ entries.

An added benefit is that the client can use the retrieved entries to both retrieve the queried entry \emph{and} to preprocess for the next phase. After an initial setup phase, the system supports continuous query refresh using its standard interaction pattern, incurring essentially no additional communication overhead. This natively incorporates the technique for uninterrupted, continuous querying introduced in PIANO-23 as "unbounded" querying. 

The pre-requisites that make these two transformations possible are as follows:
\begin{itemize}
    \item Only $\sqrt{n}$ entries are locally accessed by the server during a query
    \item No processing is needed on the server side to select entries, they are identified by the client in the clear
\end{itemize}

Therefore, these transformation can be applied to a select number of existing schemes, unlocking the potential of adapting these schemes for use in the default-server paradigm. 

\paragraph{\basename to \sysname} When applying this transformation to \basename, we get \sysname, the full  version of our scheme which works with default servers. Amortized communication remains the same, but the work performed on the server side is reduced --- the server must still retrieve and transmit each entry, but the server performs no specialized PIR techniques. The computation on the server side is replaced by the client, which XORs each incoming entry with the hint and performs batched preprocessing with the same multiset of entries. This results in the ability to run \sysname on a default server using only publicly available retrieval \verb|API|s, such as Wikidata.

\paragraph{PIANO in the default server setting.} In PIANO-23 and PIANO-24, a client query retrieves from the server parities corresponding to a redacted hint\footnote{PIANO calls it a \textit{punctured} hint.}. In PIANO-24, the server accesses $2\sqrt{n}-1$ DB entries, XORs them into $\sqrt{n}$ distinct parities, and sends them back; $\sqrt{n}-1$ of the parities are unused by the client. In PIANO-23, only $\sqrt{n}$ DB entries are accessed by the server and one parity sent back. Both variants incur an amortized communication complexity $O(\sqrt{n}\cdot\beta)$. Therefore, to adapt (either) PIANO to the default server setting, the entries accessed by the server can simply be sent back to the client for parity calculations.

\paragraph{RMS-24 in the default server setting.} The transformation for RMS-24 is very similar. Recall that in the original algorithm, the client requests from the server two redacted hints, one real hint and one dummy, each comprised of $\sqrt{n}/2$ entries. The client discards the dummy hint and XORs the real hint with the original (unredacted) hint, revealing the desired entry. In this case, the client instead requests $\sqrt{n}$ individual entries: $\sqrt{n}/2$ are the entries from the desired hint, excepting the desired entry itself, and $\sqrt{n}/2$ are entries from the remaining partitions.

\paragraph{WR-25 in the default server setting.} WR-25 improves client-side hint storage by utilizing a \textit{hint table} with $n$ entries arranged by the client in a $2n/T \times T$ matrix (usually we set $T = \sqrt{n}$). Each database entry only appears in the hint table once, yielding an optimal hint store of $O(\sqrt{n}\cdot\beta)$ bits (i.e., shaving a logarithmic factor off). 

The scheme employs a novel hint replenishment strategy: previous works discard hints after using them for querying, and therefore need to maintain spares for replenishment. In contrast, the WR-25 fetches in each query all the entries in a hint except the queried index in order to add them back to existing table rows. In this way, client-side storage is reduced, and as an unintended byproduct, the scheme can natively work with a default server, with the server performing $O(\sqrt{n})$ accesses per query and the communication cost $O(\sqrt{n}\cdot\beta)$. 

\section{Complexity Analysis}

In this section, we discuss communication, computation, and client-side storage requirements for \basename and \sysname. 
A comparison of the results with previous works is provided in Section~\ref{sec:comparison}.

\subsection{\basename}

\paragraph{Communication}
In order to pre-process hints for the database, all entries from the database must be streamed. The number of entries in the database is $n$ and the size of each entry is no greater than $\beta$, so the maximum communication requirement for preprocessing hints is always $n\beta$. Following this, for each retrieval the server computes and returns a singular redacted hint. \basename will perform $\sqrt{n}$ retrievals before the scheme must be refreshed, so $n /\sqrt{n}=\sqrt{n}$; thus, the amortized cost of preprocessing over $\sqrt{n}$ retrievals is $\sqrt{n}\cdot\beta$, and adding a single additional redacted hint of size $\beta$ results in negligible additional work, so we have $O(\sqrt{n}\cdot\beta)$ amortized communication per hint.

\paragraph{Client-side computation}
In the preprocessing phase, the client must make $m$ hints of $\sqrt{n}$ entries each, resulting in a total of $m\sqrt{n}$ XORs.

For each retrieval, the client performs a single XOR, leading to a constant client-side computation for each query.

\paragraph{Client-side storage}
The client stores a total of $m$ hints of size $\beta$, plus a seed for each hint. This results in a total client-side storage of $2m\beta$, or $O(m\beta)$. After each query (up to $\sqrt{n}$ queries), the client stores the entry which is retrieved, so the total storage required will be $O(m\beta + \sqrt{n})$.

\paragraph{Server computation}
Server computation is equal to the number of entries that must be retrieved per round. (The cost of XORs is negligible and so will not be included in this analysis.) Each hint is of size $\sqrt{n}$, so the number of entries retrieved to create the redacted hint will be $\sqrt{n} - 1$, resulting in a per-round server computation cost of $O(\sqrt{n})$. 

\subsection{\sysname}

\paragraph{Communication}

Like \basename, \sysname must initially download $n\beta$ entries in order to preprocess hints. However, in \sysname, each query must download $\sqrt{n}$ entries. This results in a \emph{per-query} communication of $O(\sqrt{n}\cdot\beta)$. The combined communication of one preprocessing phase and $\sqrt{n}$ retrievals is $2n$, and divided across $\sqrt{n}$ queries, we have an amortized communication of $2\sqrt{n}\cdot\beta$, or $O(\sqrt{n}\cdot\beta)$, identical to the per-query communication.

Additionally, the communication can be halved simply by using the entries from each query to preprocess in parallel, converting \sysname to a continuous scheme. This results in a total communication of $n\beta$ and an amortized communication of $\sqrt{n}\cdot\beta$, or $O(\sqrt{n}\cdot\beta)$.

\paragraph{Client-side computation}
In the preprocessing phase, the client must perform $\sqrt{n}$ XORs for each of $m$ hints, leading to a total client-side computation of $O(m\sqrt{n})$.

The client performs $\sqrt{n}$ XORs per query, leading to a per-query computation of $O(\sqrt{n})$.

\paragraph{Client-side storage}
Client-side storage in the default setting is identical to the cooperative-server setting.

\paragraph{Server computation}

The server computation in \sysname is nearly identical to \basename, with the exception that the server performs no XORs (as noted above, XORs are negligible). In the preprocessing phase, the server must retrieve and send $n$ entries to the client. The server will retrieve and send $\sqrt{n} - 1$ entries per query. If the scheme is continuous, the total server computation per $\sqrt{n}$ queries will be $n$, while if the scheme is not continuous, the total server computation per $\sqrt{n}$ queries is $2n$, making the amortized server computation $O(\sqrt{n})$ in both the limited-query and continuous settings.

\subsection{Comparison} 
\label{sec:comparison}

\basename and \sysname are grounded and have benefited from a long line of works in the PIR settings. In this section, we compare the complexities of \basename and \sysname with PIANO-23, PIANO-24, RMS-24, and WR-25. 

\textbf{Cooperative/Default server.} First, we note that none of the works preceding \sysname were designed to operate with a default (``non-cooperative'') server. Additionally, for normal complexity measures, we provide two dimensions of comparison: in the cooperative server setting, we position \basename along the above schemes directly. In the default server setting, we  position the full \sysname scheme against the modified/applied variants of PIANO-23, PIANO-24, RMS-24 and WR-25, discussed in Section~\ref{sec:discussion}.

\textbf{Pre-processing/per-query communication.}
In the cooperative server setting, we distinguish the actual query cost from the pre-processing cost. Here, \basename achieves the lowest known per-query communication cost of precisely $1 \cdot \beta$, improving by a factor of $2$ over the previously best known solution, RMS-24\footnote{PIANO-23 has communication cost $1 \cdot \beta$ but it does not support arbitrary query workloads. While this additional restriction may be appropriate in certain contexts, it does not generally hold and departs from the standard PIR model.}.
    
In the default server model, \sysname and all modified/applied methods exhibit $O(\sqrt{n} \cdot \beta)$ per-query communication cost. This allows all of them to stream entries and carry the pre-processing continuously at no added complexity.  
    
\textbf{Per-client storage.}
    \basename matches the client-side storage of PIANO-23, PIANO-24 and RMS-24 of $O(\sqrt{n} \cdot \log{n} \cdot \beta)$. 
    WR-25 has $O(\sqrt{n} \cdot \beta)$ client-side store (i.e., a $\log{n}$ is shaved off) but per-query communication cost $O(\sqrt{n} \cdot \beta)$.

    The hint structure and corresponding client-side storage of \sysname is the same as \basename, and likewise, remain the same in all the other schemes when modified/applied to the default server setting. 
    
\textbf{Work/Computation.}
    The server work is dominated with DB accesses, which is $O(n)$ for pre-processing and $O(\sqrt{n})$ per query in all methods. The client-side pre-processing computation is $O(n\log(n))$ in all schemes. Searching for a hint containing a queried index incurs $O(n)$ computation in \sysname, compared with $O(\sqrt{n})$ computation in all above mentioned systems due to sharding. While for most practical database sizes this computation complexity is not a concern (order of seconds), it is left as an open question to reduce the client-side computation.

\section{Evaluation}

\subsection{Experimental setup}

All experiments are run on a local server equipped with a 3.65 GHz AMD EPYC 9135 16-Core Processor and 30GB of RAM, running Ubuntu 24.04.04 LTS. Online hint search is parallelized across all available threads (32). The offline phase parallelizes hint construction across all available threads, with the database placed in shared memory to avoid per-worker copies.

The scheme is evaluated on databases with $2^{20}$, $2^{24}$, and $2^{28}$ entries with entry sizes varying from 64KB to 4GB. Parameters are set to $k = \lceil \sqrt{n} \rceil$ and $m = \lceil 2C \ln(n) \cdot n/k \rceil$ with coverage constant $C=4$, yielding a per-entry covering probability of at least $1 - 1/n^{1.25}$ by applying the Chernoff bound. For each configuration, $50$ queries are issued to unique, uniformly random targets, and the mean per-query online latency is reported. 

We simulated network latency based on real Internet transmission latency data taken from Cloudflare \cite{cloudflare-radar-bandwidth-by-continent-worldwide}.

\paragraph{Cooperative-server PIR.} 

In the cooperative paradigm, the server computes the XOR of the redacted entries on the client's behalf and returns a single $\beta$-byte response.
This is compared against RMS-24, the previously known paradigm whose per-query communication $O(1)$ is asymptotically optimal.

\paragraph{Default-server PIR.}
To validate our approach in the default-server setting, \basename was adapted accordingly by having the server return the $(k-1)$ redacted entries individually, which the client then XORs locally, yielding a per-query download of $(k-1)\beta$ bytes.  
To demonstrate its suitability for real-world web servers, the approach was applied to Wikidata. Rather than downloading and querying the entire database---which would be prohibitive in practice---the \sysname implementation queries on a subset of Wikidata. This implementation is available in our codebase.

\subsection{Implementation Details}
\sysname and \basename are implemented in a shared codebase. The two variants differ only in the server-side response: in \basename, the server returns the XOR of the redacted entries, whereas in \sysname, the server returns those entries individually and the client performs the XOR locally.

Benchmarking code was implemented alongside these schemes for \basename and RMS-24. The total benchmark implementation simulates network latency and server retrieval speeds using simple calculations. Seed expansion, hint search, and XOR times were emulated with randomized data.

\paragraph{Hint representation.}
Each hint is stored compactly as a seed together with its precomputed XOR value. In addition, the implementation stores a designated random replacement entry for each hint in order to support the local maintenance step described in Section~\ref{subsec:entryCaching}. Thus, each hint carries exactly the information needed to reconstruct its multiset, recover the queried entry, and update the remaining hint multiset after use.

\paragraph{Seed generation.}
Hint seeds are generated in counter mode from a 64-bit master seed, yielding a deterministic stream of distinct candidate seeds. Each candidate seed is
expanded via Algorithm~\ref{alg:expand-multiset-seed} to produce a size-$k$ multiset. Because distinct seeds may still expand to the same multiset, though it is a negligible probability, the implementation detects and discards duplicate expansions during preprocessing. This ensures that each stored hint is unique, so security is not compromised by querying the multiset more than once. Counter-mode generation also makes the process naturally resumable, which is useful for the continuous preprocessing variant discussed in ~\ref{subsec:continuous_preproc}

\paragraph{Parallelism.}
Hint construction is parallelized across worker processes. Each worker handles a disjoint portion of the seed space and computes the corresponding hint XOR values.
The online phase parallelizes hint search across worker processes, assigning a number of seeds to each worker and stopping all threads once a suitable candidate is found.

\paragraph{Network and server latency simulation.}
Network latency is calculated as
\[
\text{network\_ms} = \frac{(\text{upload\_bytes}+\text{download\_bytes}) \cdot 8}{\text{bits\_per\_ms}}.
\]
For \basename, the number of bytes downloaded is $1 \times \beta$, whereas for RMS-24 it is $2 \times \beta$.

We generously assume a server network speed is as high as raw I/O, and take service\_ms to be $\sqrt{n}\,\beta$ divided by IO throughput. To model wait times under concurrent load, we adopt an M/M/1 queuing model, assuming that each client issues queries at a rate of $1/\text{base\_ms}$, where base\_ms denotes the total round-trip time under no contention. The resulting server utilization is
\[
\rho = \frac{\text{num\_clients} \cdot \text{service\_ms}}{\text{base\_ms} + \text{service\_ms}}
\]
and the corresponding mean wait time is $w = (\text{service\_ms} \cdot \rho)/(1 - \rho)$.
The server is said to be \emph{saturated} when $\rho \ge 1$, indicating that requests arrive more quickly than the queue can be drained.

Server-side XOR latency is negligible, usually pipelined alongside memory/storage I/O, and so it is excluded.

\subsection{\basename Evaluation Results}

\begin{figure*}[t]
    \centering
    \includegraphics[alt={A plot illustrating total time for a single query vs size of beta when n = 2^20}, width=.33\linewidth]{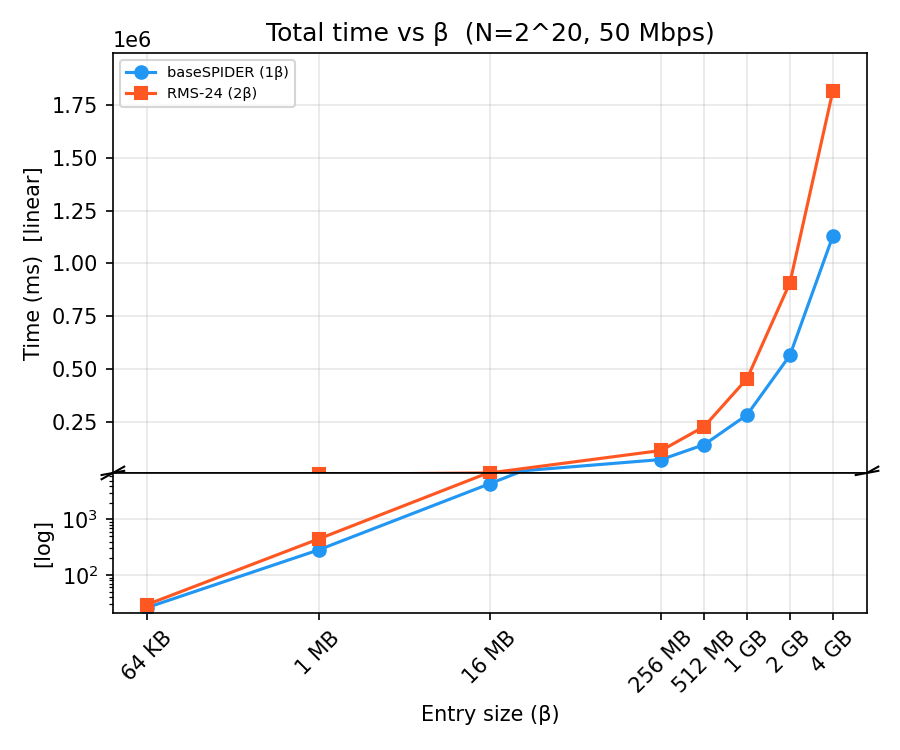}
    \includegraphics[alt={A plot illustrating total time for a single query vs size of beta when n = 2^24}, width=.33\linewidth]{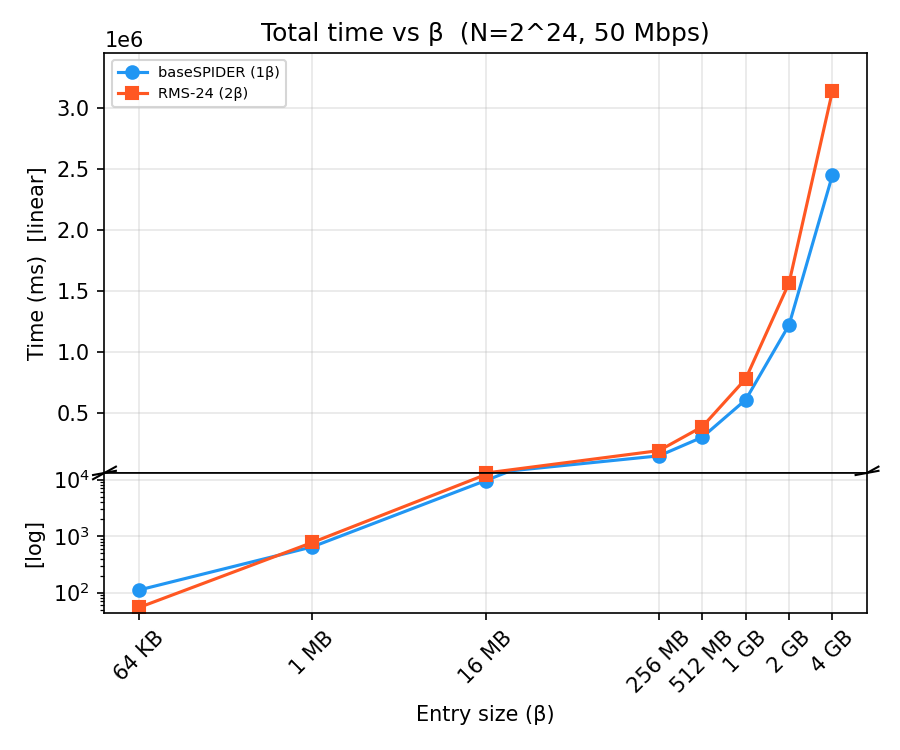}
    \includegraphics[alt={A plot illustrating total time for a single query vs size of beta when n = 2^28}, width=.33\linewidth]{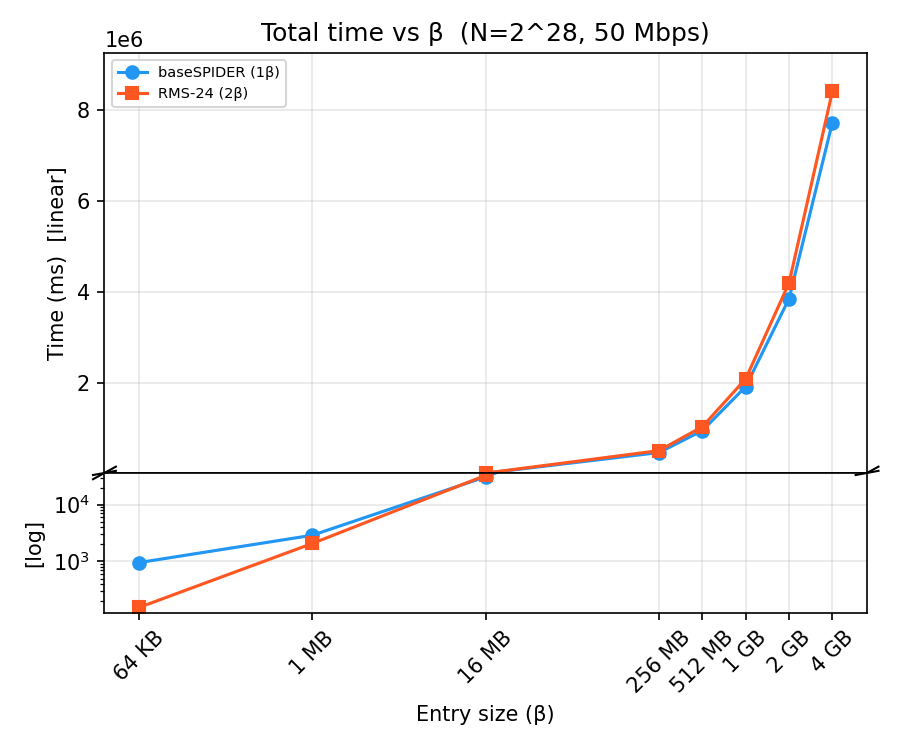}
    \includegraphics[alt={A plot illustrating total time for a single query vs size of beta when n = 2^20}, width=.33\linewidth]{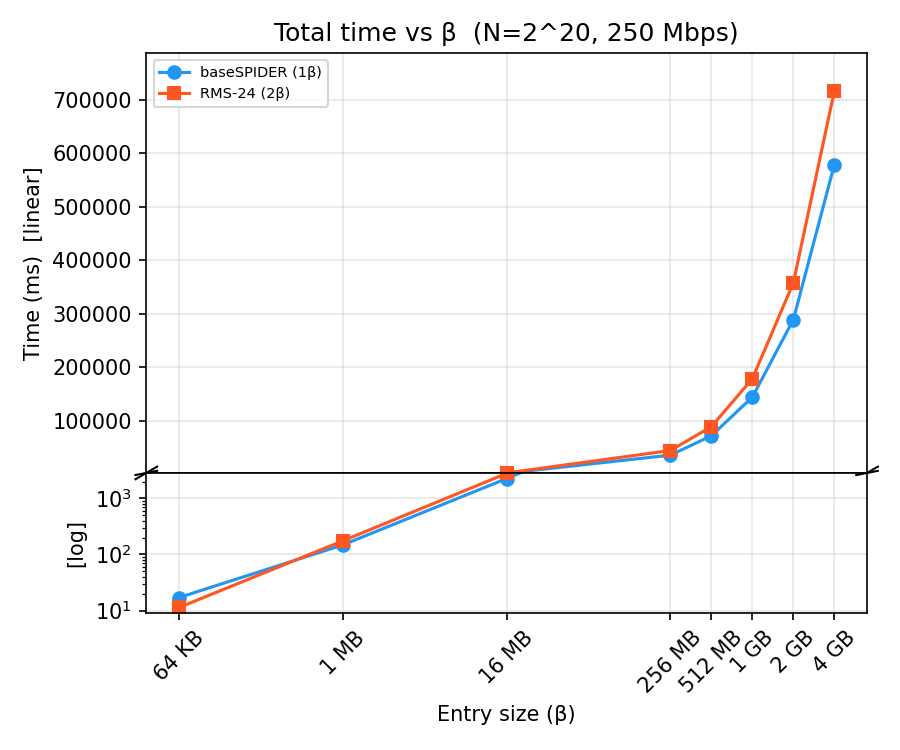}
    \includegraphics[alt={A plot illustrating total time for a single query vs size of beta when n = 2^24}, width=.33\linewidth]{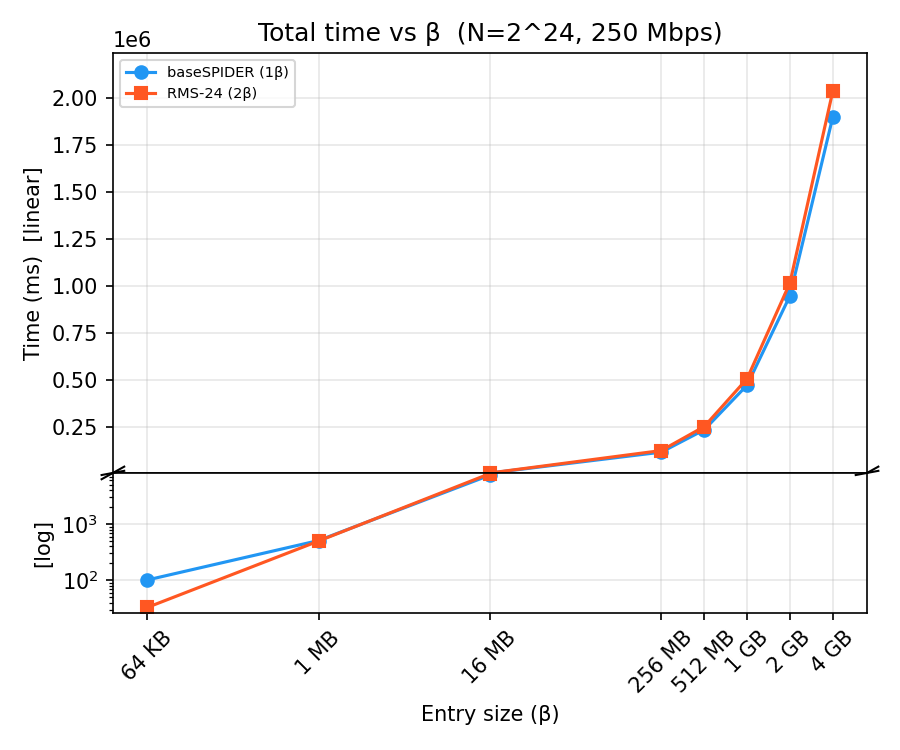}
    \includegraphics[alt={A plot illustrating total time for a single query vs size of beta when n = 2^28}, width=.33\linewidth]{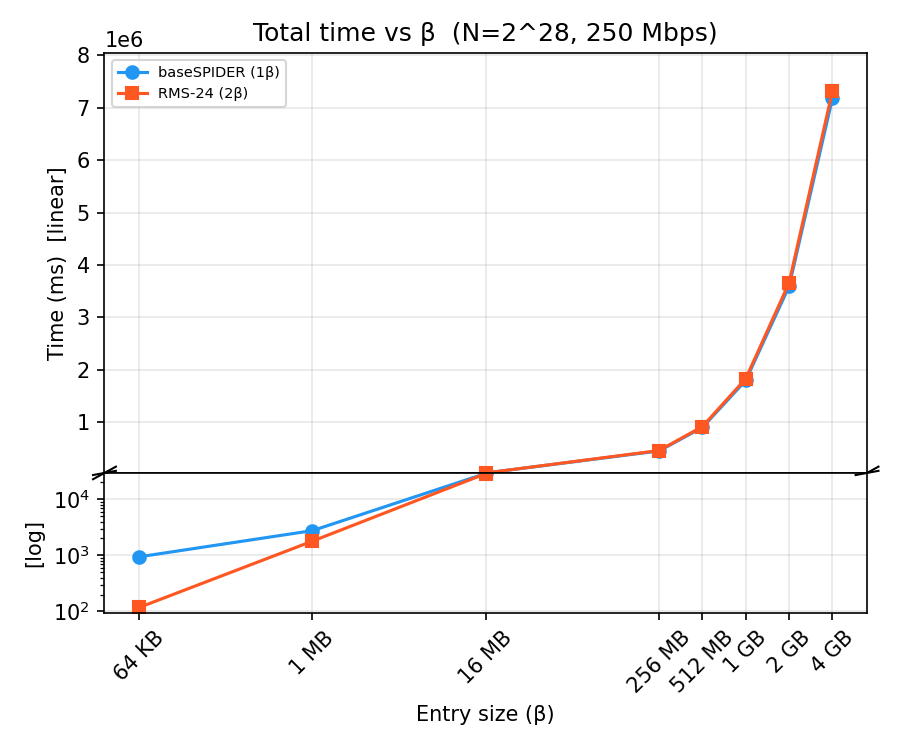}
    \caption{Plots illustrating total time in milliseconds versus various values of $\beta$. We can see as the size of $\beta$ increases, \basename gains an advantage over RMS-24 due to the two half-size hints RMS-24 returns.}
    \label{fig:time_vs_beta}
\end{figure*}

\begin{table}
  \centering
  \begin{tabular}{lccc}
    \toprule
    \textbf{SCHEME} & \textbf{$2^{20}$} & \textbf{$2^{24}$} & \textbf{$2^{28}$} \\
    \midrule
    RMS-24 & 0.087 & 0.268 & 1.045 \\
    \basename & 3.192 & 46.236 & 570.093 \\
    \bottomrule
  \end{tabular}
  \caption{Hint search latency in ms.}
  \label{tab:results}
\end{table}

In this subsection, we report the evaluation results for \basename. The evaluation focused on the time in milliseconds required to serve a single query, so our calculations do not address preprocessing or refreshing. We focus on three aspects in particular and their effects on a practical system: (1) Hint search and expansion, (2) network latency, and (3) server throughput.

\paragraph{Hint search and expansion.}
\basename requires a linear hint search in order to find a multiset which contains a desired item. This differs from RMS-24 and other schemes which use a sharded database, notably that \basename must expand all seeds to $k$ indices and search every item as opposed to generating and comparing just one item per seed (since sharding allows the client to compare only indices within that shard to the desired index). Table~\ref{tab:results} details the differences for the database sizes in our experiments.
While this is a disadvantage when retrieving small entries from large databases, as shown below, the time required for hint search is quickly dwarfed by network bandwidth and database retrieval limitations as entries become larger.

\paragraph{Network latency.}
When a single client is accessing the server, network latency quickly becomes the dominant feature. For this reason, the ability to transmit half the data of RMS-24 becomes an advantage, resulting in a measurable difference that becomes larger as $\beta$ grows. 
Figure~\ref{fig:time_vs_beta} shows single client end-to-end latency plotted against different values of $\beta$ and network speeds varied between two representative values, 50 Mbps and 250 Mbps~\cite{cloudflare-radar-bandwidth-by-continent-worldwide}. 
The figure shows that \basename has a distinct advantage over RMS-24 in both. 

\begin{figure*}
    \centering
    \includegraphics[alt={A plot illustrating total latency per query vs throughput of the server when n = 2^20 and network speed is 50Mbps}, width=.33\linewidth]{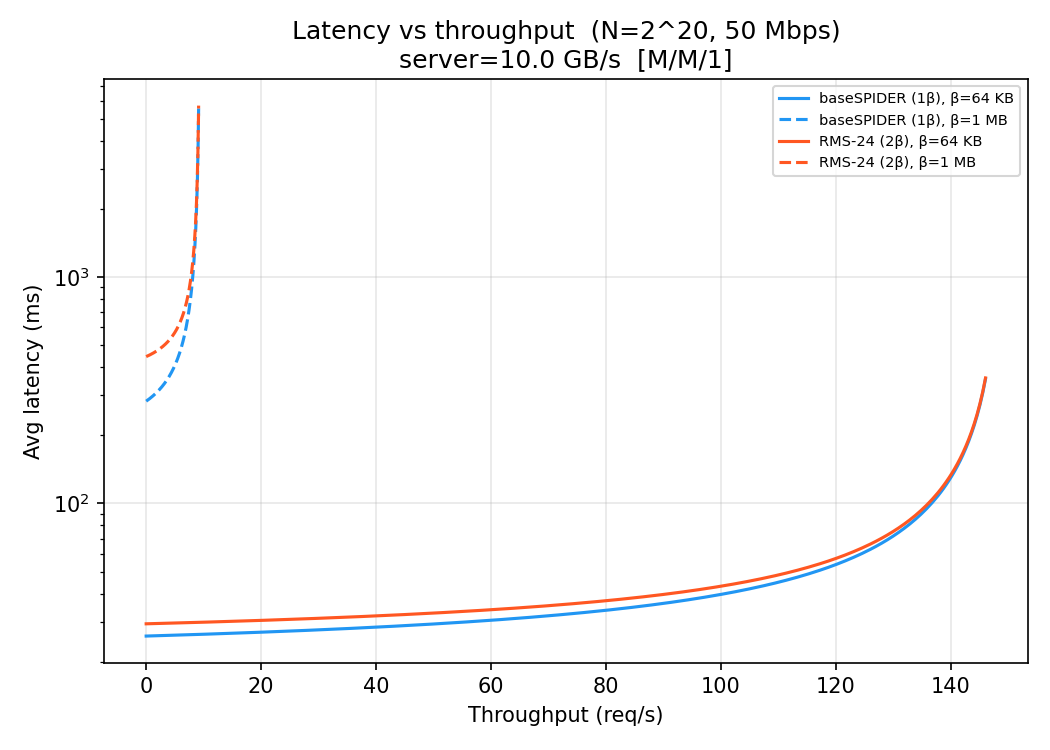}
    \includegraphics[alt={A plot illustrating total latency per query vs throughput of the server when n = 2^24 and network speed is 50Mbps}, width=.33\linewidth]{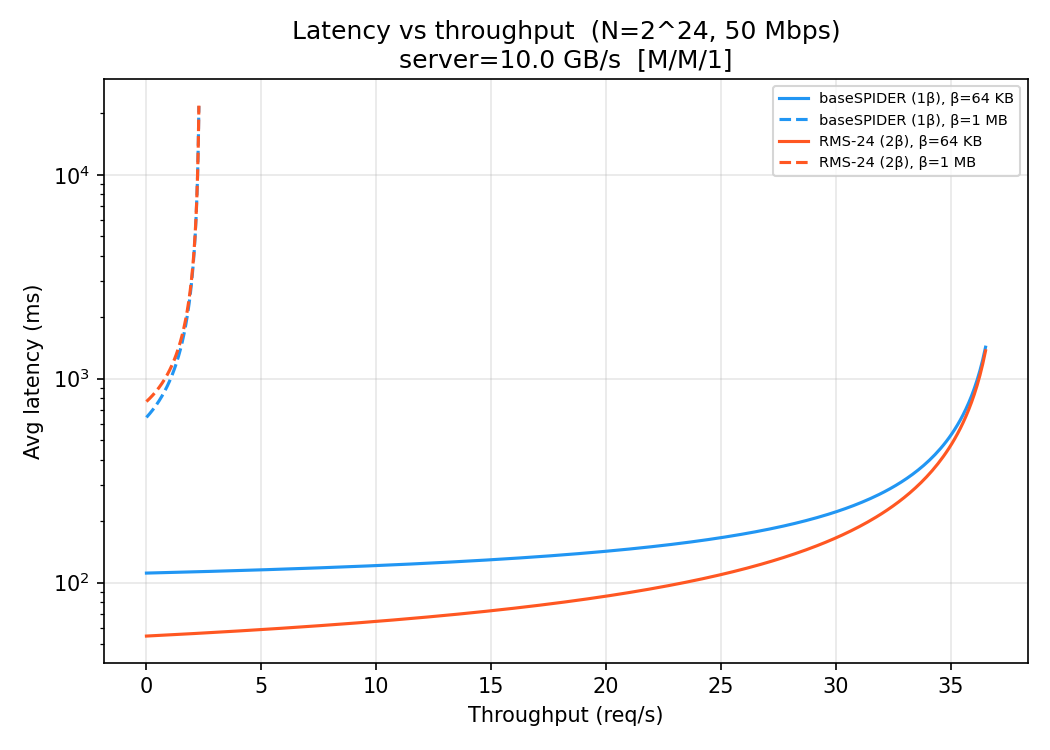}
    \includegraphics[alt={A plot illustrating total latency per query vs throughput of the server when n = 2^28 and network speed is 50Mbps}, width=.33\linewidth]{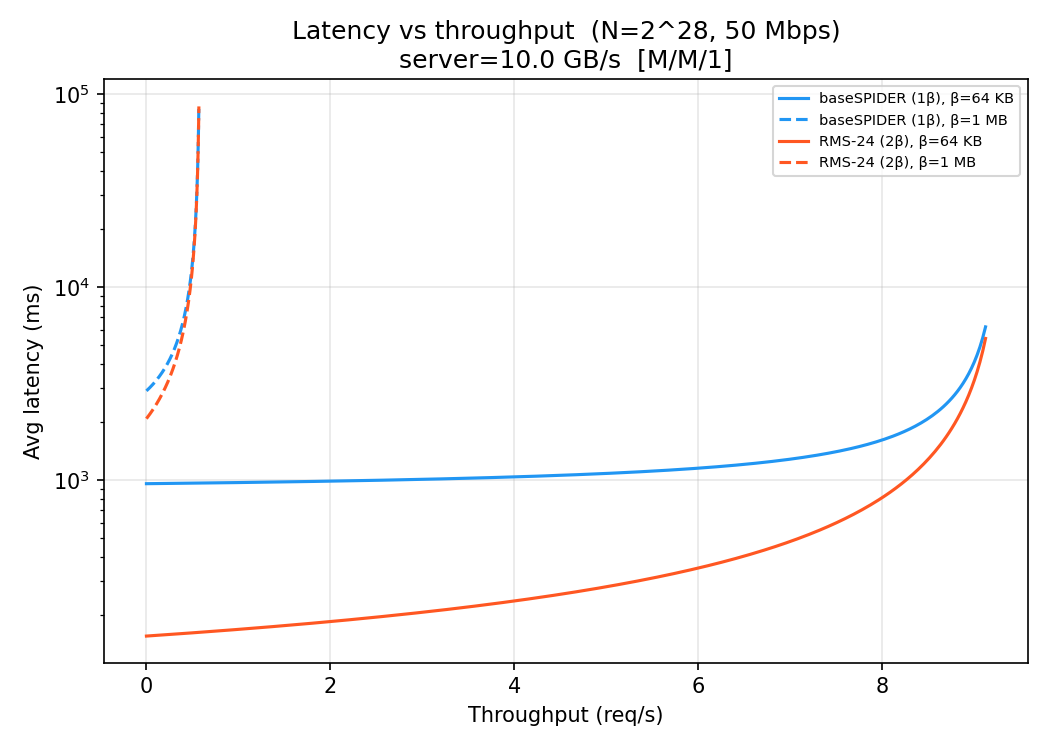}
    \includegraphics[alt={A plot illustrating total latency per query vs throughput of the server when n = 2^20 and network speed is 250 Mbps}, width=.33\linewidth]{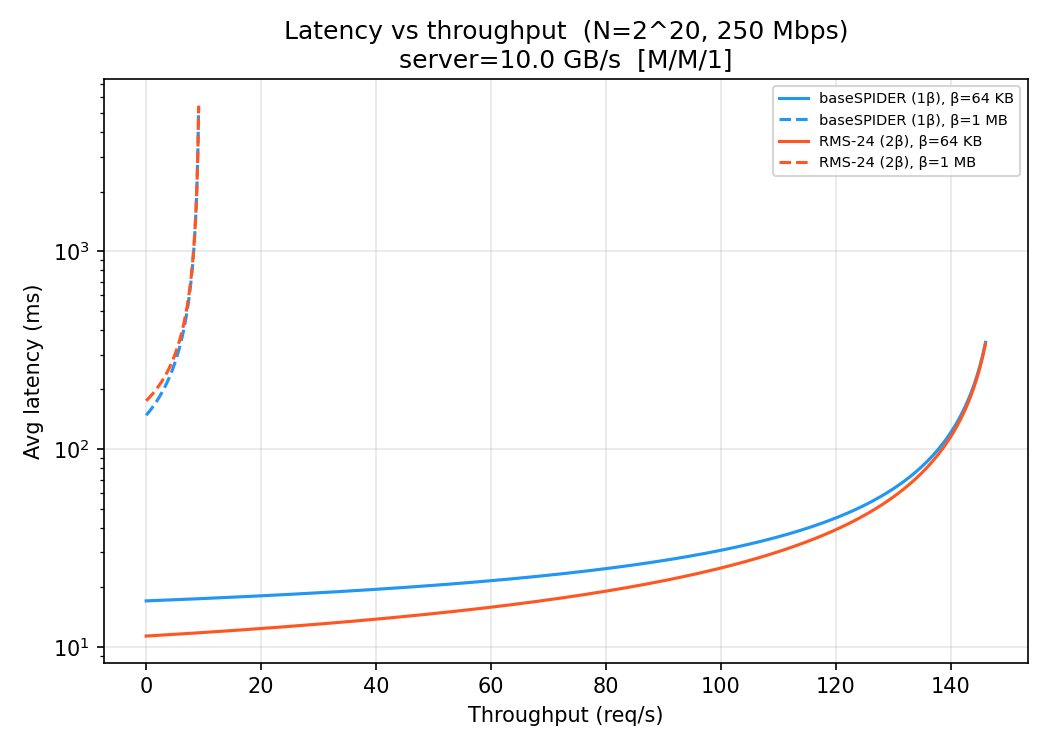}
    \includegraphics[alt={A plot illustrating total latency per query vs throughput of the server when n = 2^24 and network speed is 250 Mbps}, width=.33\linewidth]{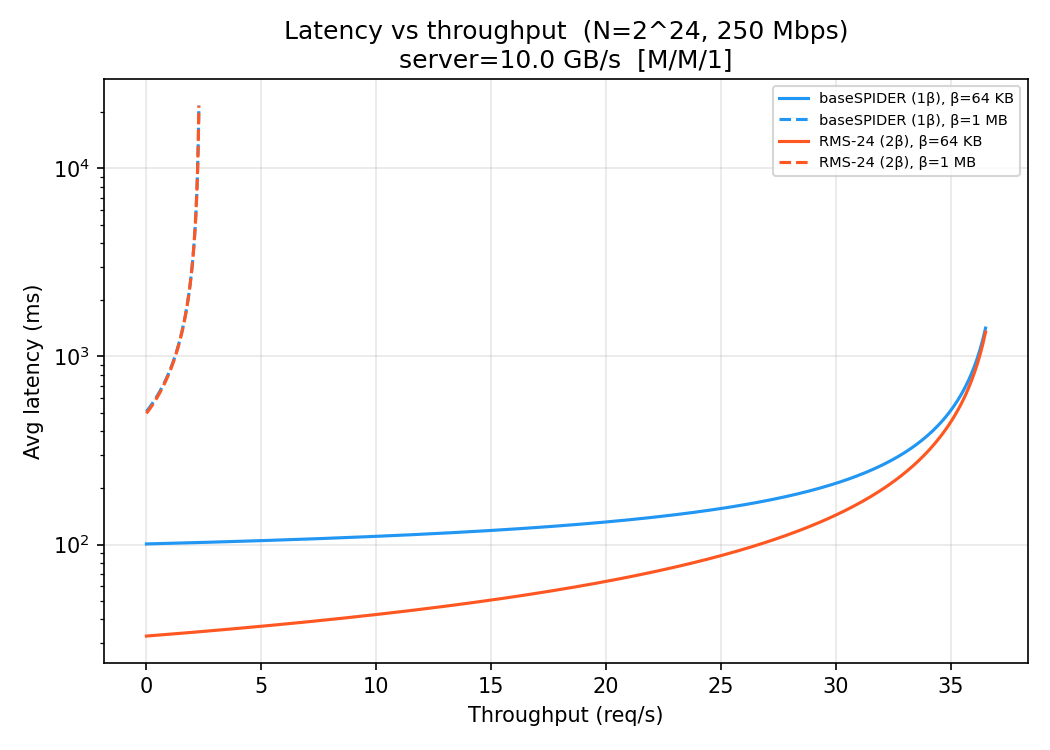}
    \includegraphics[alt={A plot illustrating total latency per query vs throughput of the server when n = 2^28 and network speed is 250 Mbps}, width=.33\linewidth]{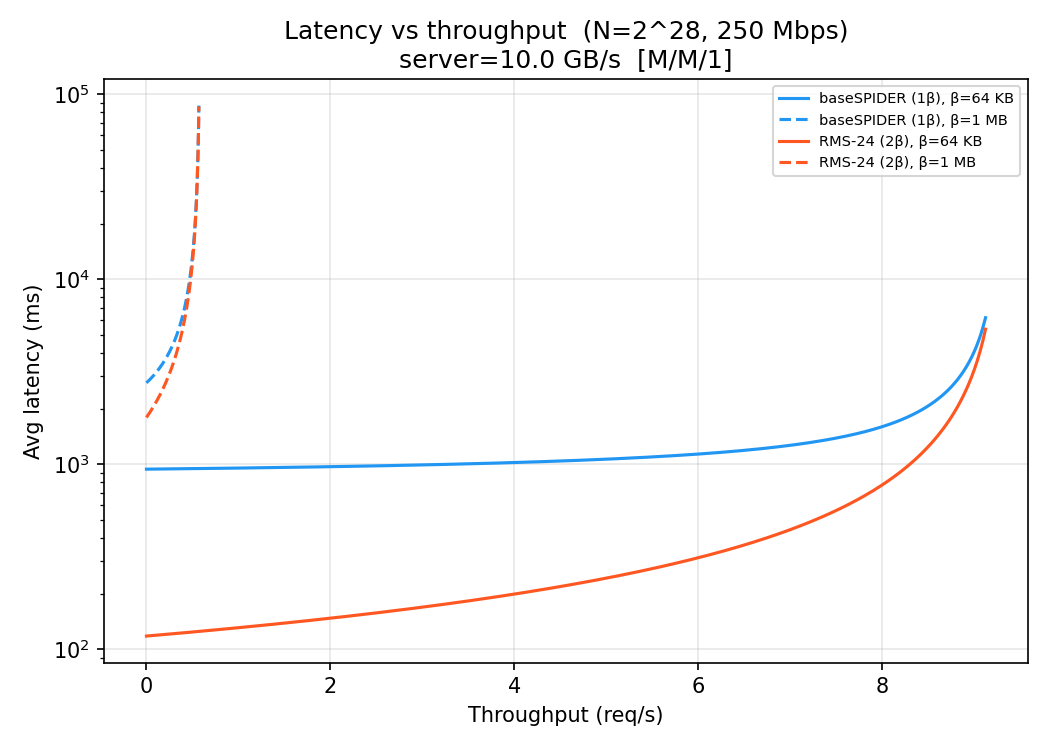}
    \caption{Plots illustrating latency in milliseconds versus throughput of the server. At a relatively low throughput, the latency for each request approaches infinity.}
    \label{fig:latency_vs_tput}
\end{figure*}

\paragraph{Server throughput.}
Even when database entries are on the smaller size, every individual server at the back-end of a web service typically need to cater to multiple concurrent clients. Figure~\ref{fig:latency_vs_tput} shows client latency as server throughput is driven up until saturation (at which point client wait times spike up and the server chokes). Despite the complexity incurred by hint search, \basename latency is comparable with RMS-24 at the same throughput rates. It exhibits slightly lower latency for similar throughput points compared with RMS-24 when fewer entries need to be searched. The advantage is retained when the ratio between entry size and network speed is higher, and vanishes when the network bandwidth increases. 

\subsection{WikiData Integration}

To demonstrate both the relevance of the default-server setting and the applicability of \sysname beyond synthetic row-indexed databases, \sysname has been integrated with the WikiData SPARQL endpoint.

\paragraph{Why the default-server setting fits.}
WikiData returns raw entries and does not provide any primitive for combining multiple rows on the server side. This matches the assumptions of \sysname: the server returns the requested entries individually, and the client performs
the XOR locally. In contrast, the cooperative paradigm used would require server-side
functionality that WikiData does not expose. 

\paragraph{Key space and representation.} Since WikiData is keyed by entity labels or identifiers rather than by integer row indices, using it with \sysname requires an offline enumeration of the relevant key space together with a deterministic mapping from keys to indices.
During the offline phase, the client issues a SPARQL query to enumerate a fixed
set of relevant WikiData entities and collect their associated values. This information is used only for local preprocessing: it allows the client to construct its hints and define the finite database over which the protocol
operates.
The online phase still queries the original WikiData service directly.

\paragraph{Key-to-index mapping.}
To allow the client to query by label rather than by row index, the offline snapshot is equipped with a minimal perfect hash over the enumerated key set \cite{DBLP:journals/corr/abs-1910-06416} \cite{https://doi.org/10.4230/lipics.sea.2025.21} \cite{DBLP:journals/corr/abs-2104-10402}. This yields a deterministic key-to-index map with no collisions on the chosen keys, allowing both preprocessing and online queries to resolve the same label to the same database index in constant time. Minimal perfect hashing is a standard tool for compact static key indexing.
The main practical requirement is consistency between the offline snapshot and the values used during online retrieval. Accordingly, for the purpose of index calculation, the integration assumes the key-set is fixed throughout preprocessing and query execution.
\section{Related Works}

Private Information Retrieval (PIR) has been extensively studied across single-server, multi-server, and preprocessing-based settings. Classical PIR constructions achieve sublinear communication either by relying on multiple non-colluding servers or heavy cryptographic tools, often incurring prohibitive server computation costs. 
Works on doubly-efficient PIR and FHE-based constructions explore regimes with strong asymptotic guarantees~\cite{linDoublyEfficientPrivate2023,luoFasterFHEBasedSingleServer2024}. Multi-server PIR and information-theoretic constructions provide improved asymptotics under non-collusion assumptions~\cite{singhInformationTheoreticMultiserverPrivate2025,ghoshalZeldaEfficientMultiserver2025}, while systems-oriented works extend PIR to broader applications such as private search and secure communication~\cite{zhouPacmannEfficientPrivate2024,ahmad2021addra}. However, these approaches rely on strong deployment assumptions, including multiple servers, specialized APIs, or cooperative infrastructure.

Recent advances focus on \emph{preprocessing PIR}, where a stateful client stores hints to reduce online complexity. Foundational results show that preprocessing enables overcoming classical lower bounds on server computation and communication, yielding schemes with sublinear amortized costs~\cite{cryptoeprint:2022/081_CGHK,cryptoeprint:2019/1075_CGK,lazzarettiNearOptimalPrivateInformation2022,zhouOptimalSingleServerPrivate2023}. Foundational works characterize achievable trade-offs in preprocessing PIR~\cite{cryptoeprint:2024/976}, and exploration of stateful PIR abstractions formalize the role of client memory in reducing online costs~\cite{10.1007/3-540-44598-6_4,10.1145/3243734.3243821}. Subsequent works achieve tight space-time bounds and improved asymptotics~\cite{wangSingleServerClientPreprocessing2025,hooverPlinkoSingleServerPIR2025}. Other works focus on simplifying constructions and improving concrete efficiency, including practical amortized schemes with low online communication overhead~\cite{RenLing2024SaPA,GhoshalAshrujit2024EPPW,cryptoeprint:2023/452_PIANO_11_12}.

baseSPIDER fits within the stateful PIR paradigm, building on the insight that client preprocessing can amortize work across queries. Compared to prior preprocessing schemes~\cite{cryptoeprint:2022/081_CGHK,zhouOptimalSingleServerPrivate2023,wangSingleServerClientPreprocessing2025,hooverPlinkoSingleServerPIR2025}, it matches optimal asymptotic communication while improving constant factors and offering a conceptually simpler construction. These improvements are particularly impactful in regimes with large database entries, where constants dominate practical performance.

Early steps on \emph{non-cooperative} or \emph{default-server} PIR settings, where the server exposes only a standard indexed interface, include systems such as Zeal~\cite{zipkin:zeal}, which aim to support PIR over unmodified databases but rely on trusted proxies. 
To our knowledge, \sysname is the first construction achieving single-server PIR in this fully non-cooperative setting. While related efforts such as Zeal~\cite{zipkin:zeal} explore similar goals, SPIDER is obtained via a simple transformation from baseSPIDER, removing server cooperation entirely, enabling practical deployment over existing unmodified systems.

\end{doublespacing}

\renewcommand\refname{References}
\bibliographystyle{plainnat} 
\textnormal{\bibliography{refs.bib}}
\newpage

\end{document}